\documentclass[letterpaper, 10 pt, conference]{ieeeconf}

\IEEEoverridecommandlockouts
\usepackage{amsmath,amssymb,amsthm}
\usepackage{mathtools}

\usepackage{microtype}

\usepackage{graphicx}
\usepackage{booktabs}
\usepackage{tikz}
\usepackage{pgfplots}
\pgfplotsset{compat=1.18}
\usetikzlibrary{arrows.meta,calc}

\usepackage[
backend=biber,
style=ieee,
natbib
]{biblatex}
\AtEveryBibitem{%
    \clearfield{doi}
    \clearfield{issn}
}
\DeclareSourcemap{
  \maps[datatype=bibtex]{
    \map{
      \step[fieldset=urldate, null]
    }
  }
}

\usepackage{hyperref}
\usepackage{xcolor}
\hypersetup{
    colorlinks=true,
    linkcolor=blue!70!black,
    citecolor=red!70!black,
    urlcolor=blue!70!black
}

\newtheorem{theorem}{Theorem}[section]
\newtheorem{lemma}[theorem]{Lemma}

\newtheorem{corollary}[theorem]{Corollary}

\theoremstyle{definition}
\newtheorem{definition}[theorem]{Definition}

\theoremstyle{remark}
\newtheorem{remark}[theorem]{Remark}

\newcommand{\E}{\mathbb{E}}

\newcommand{\poa}{\mathrm{PoA}}

\DeclareMathOperator*{\argmax}{arg\,max}

\title{\LARGE \bf Invariant Price of Anarchy and Multiplicative Smoothness}

\author{Ilia Shilov, Heinrich H.\ Nax, Saverio Bolognani%
\thanks{I. Shilov and S. Bolognani are with the Automatic Control Laboratory, ETH Zurich, 8092 Zurich, Switzerland. Email: \{ishilov, bsaverio\}@ethz.ch}%
\thanks{H. H. Nax is with the Zurich Center for Market Design, University of Zurich, 8050 Zurich, Switzerland. Email: heinrich.nax@uzh.ch}%
\thanks{This work was supported by the Swiss National Science Foundation under the NCCR Automation (grant agreement 51NF40\textunderscore225155).}%
}

\begin{document}
\maketitle
\thispagestyle{empty}
\pagestyle{empty}

\begin{abstract}
The Price of Anarchy~(PoA) is a popular measure of the costs of decentralization in terms of efficiency losses. Almost all PoA analyses operate within a framework assuming both Cardinal Full-Comparability~(CFC) and smoothness, in which case any derived bounds conveniently extend beyond pure Nash to coarse correlated equilibria and no-regret learning outcomes. However, interpersonal utility comparability is an additional assumption that generally has to be justified. Without it, cardinal utilities (e.g. defined under classical von~Neumann--Morgenstern framework) are unique only up to agent-specific affine transformations, rendering both the utilitarian PoA and the classical smoothness conditions representation-dependent. In this paper, we operate under a more general Cardinal Non-Comparability~(CNC) framework, under which the weighted Nash welfare is a canonical admissible aggregator. We introduce \emph{multiplicative smoothness}, a product-form condition matched to the multiplicative structure of Nash welfare, and obtain PoA bounds that are CNC-invariant and extend to coarse correlated equilibria. We demonstrate applicability of our framework on single-choice welfare games, deriving the bounds through simple proof relying on multiplicative retention envelope and geometric closure. The interpretation of this bound in terms of the true cost of decentralization depends crucially on interpersonal comparability of utilities.
\end{abstract}

\section{Introduction}\label{sec:Introduction}

The main point of this paper is that there can be no full understanding concerning the true cost of decentralization without interpersonal comparability of utilities. Whatever bounds can be obtained need to be interpreted given the specific structure of the application at hand and in light of whatever is known about utilities. But let us unpack the statements we just made beginning with the cost of decentralization, then relevant structural frameworks, and finally interpersonal utility comparability.

As game theory has become the standard language to model many applications to socio-technical systems, including transportation networks, energy grids, and distributed routing protocols, etc., the realization has sunk in that strategic behaviors of self-interested agents may lead to system-level inefficiencies. To measure that efficiency loss the Price of Anarchy~(PoA)~\citep{koutsoupias1999worst,roughgarden2002poa,roughgarden2004bounding} has become the standard metric: given a game together with an equilibrium concept and a social welfare function~(SWF), the PoA measures the worst-case ratio between the social welfare at equilibrium and at the social optimum. A lot of research has gone into analysis of tight and reliable PoA bounds, which has been extremely influential in shaping our understanding of the costs of distributed control. PoA is used both as an evaluation metric and as a design target in the sense of an \emph{operational objective}, where agents' utility functions are designed so as to minimize PoA and thus realigning self-interested behaviors with system-level goals~\citep{paccagnan2020utility,chandan2019optimal,Paccagnan2022Utility,strong_poa_ferguson}.

A large literature has developed analytical techniques for bounding and optimizing the PoA for various classes of games \citep{roughgarden2002poa,paccagnan2020utility, poa_markov, marden2013distributed}, the most widely adopted framework being the smoothness framework~\citep{roughgarden_robustness} as recently generalized in \cite{chandan2024methodologies}.
This framework turns the derivation of PoA bounds into the verification of the \emph{(generalized) smoothness inequality} with robustness as its main appeal: any bound proved through a smoothness argument extends from pure to mixed Nash, to correlated and to coarse correlated equilibria~\citep{roughgarden_robustness}. This natural extension is relevant because these broader equilibrium notions are guaranteed to exist under mild conditions and, in the case of coarse correlated equilibria, arise naturally as the long-run outcomes of no-regret learning dynamics~\citep{blum2007learning}.

An important aspect of PoA analysis when used for bounding randomised equilibrium notions is that payoff comparisons are essentially made in expectation under lotteries over action profiles, which naturally evokes the standard von~Neumann--Morgenstern~(vNM) model of preferences over lotteries \citep{vonneumann1944theory}. Generally, under this framework, each agent's utility representation is unique only up to \emph{individual positive affine transformations} $u_i \mapsto \alpha_i u_i + \beta_i, \alpha_i > 0.$
This fact is critical because, while the equilibrium sets of pure, mixed, correlated, and coarse correlated equilibria are \emph{invariant} under such transformations~\citep{tewolde2024game}, the standard utilitarian PoA is not; utilitarian being when welfare is measured by the sum of utilities. Hence, a player-specific affine transformation preserves the equilibrium behaviour, but may change both the social optimum and the utilitarian welfare attained at equilibrium. As a result, the inefficiency expressed by PoA may therefore actually depend not only on the strategic incentives, but also on perhaps sometimes arbitrary choices of units and baselines. In this sense, the usual utilitarian PoA can be seen as potentially misaligned with the representational freedom presupposed by the equilibrium analysis.
What is more, the same problem may appear on the bounding methodology side as the (generalized) smoothness inequalities are generally not invariant under individual affine transformations either. Thus, both the welfare benchmark and the main proof technique of modern PoA analysis inherit an implicit representation dependence that the equilibrium concepts themselves do not.

Since equilibrium behavior is invariant under individual affine transformations, the source of representation-dependence must lie in the welfare criterion used to evaluate it. This raises a natural question: \emph{which social welfare function is admissible when individual utility representations are non-unique up to player-specific affine transformations?} To answer it, we adopt the welfarist approach from social choice theory~\citep{roberts_interpersonal_1980,Sen1979_welfarism,shilov2025welfare}, which classifies admissible welfare criteria by the degree of interpersonal comparability that they presuppose as is expressed through the set of admissible affine, common or individual, transformations. A general and relevant comparability assumption is \emph{Cardinal Non-Comparability}~(CNC), under which neither the scale nor the baseline of one agent's utility is comparable with another's. This matches the informational content of vNM utility theory when individual positive affine transformations are taken as behaviorally irrelevant, for example when intensities of self-reported preferences are assumed to be idiosyncratic. Under CNC, the canonical admissible welfare aggregator is the Nash social welfare function~\citep{roberts_interpersonal_1980}. The resulting PoA is invariant to player-specific affine transformations and therefore depends only on the underlying strategic structure of the game.

Representation-dependence has immediate practical consequences. Whenever PoA is used as an objective in mechanism design, utility design, or tolling~\citep{paccagnan2020utility,chandan2019optimal,Paccagnan2022Utility}, if the metric being optimised depends on the chosen payoff representation, then the ``optimal'' design inherits the same fragility. Recent work has already highlighted useful invariance properties of PoA-related methodologies. In load balancing, \citet{bilo2022nash} adopt Nash social welfare for its \emph{player-specific scale-independence} and derive tight PoA bounds for both weighted and unweighted cases. In a different direction, \citet{chandan2024methodologies} show that PoA and generalized PoA (defined through generalized smoothness) are invariant under a common positive rescaling and player-specific offsets, which in welfarist terms corresponds to a Cardinal Unit Comparability~(CUC) assumption. In this work, we consider invariance under fully player-specific positive affine transformations, which is the natural representation class once expected utility over lotteries is admitted. This leads us to revisit the social welfare criterion itself, and then to develop the matching smoothness based bounding methodology.

Building on the welfarist framework developed in~\citep{shilov2025welfare,shilov2025ipoa}, this paper addresses these issues with the following contributions, which are intended to contribute to a better understanding of PoA without assuming interpersonal (full-)comparability of utilities.

\emph{Contributions.} (i)~We formalise a CNC-invariant notion of PoA based on Nash welfare and baseline surpluses, resolving unit and offset ambiguity in the evaluation of equilibrium inefficiency (Section~\ref{sec:framework}). (ii)~We introduce \emph{multiplicative smoothness}, a product-form condition suited to the multiplicative structure of Nash welfare, and derive the corresponding PoA bounds for both pure Nash equilibria with general parameters $(\Lambda, \lambda)$ and coarse correlated equilibria with $(\Lambda, 1)$. As in the classical smoothness framework, this represents an extension to broad randomized equilibrium classes and characterization of the long-run efficiency of decentralised learning outcomes (Section~\ref{sec:smoothness}). (iii)~Using retention envelopes and concave closure, we derive explicit PoA bounds for weighted single-choice welfare games \cite{marden2013distributed} with power law utilities ($2^p$ for degree~$p$, recovering~\citet{bilo2022nash} for cost-minimization games and extending their result to CCE) via a simple proof (Section~\ref{sec:envelopes}).

\section{Framework}\label{sec:framework}
\subsection{Game model and equilibrium}

We consider a system with heterogeneous agents, modeled as a strategic game with players \(N = \{1, \dots, n\}\), strategy sets \(S_i\), and joint strategy space \(S = \prod_i S_i\). Each player \(i \in N\) has a utility function \(u_i : S \to \mathbb{R}\) that she seeks to maximise. We fix a \emph{baseline profile}~\(s^0 \in S\) representing the outcome in which no resource is attained (e.g., the zero-resource fallback in distributed welfare games~\cite{marden2013distributed}). Each agent's utility is evaluated as a \emph{surplus} over this baseline:
\begin{equation}
\pi_i(s) := u_i(s) - u_i(s^0) > 0
\qquad
\text{for all } i \in N,\ s \neq s^0 .
\label{eq:surplus}
\end{equation}
Positivity is natural in resource allocation: every feasible allocation is strictly preferred to getting nothing. The shift does not affect best responses and hence preserves all equilibrium sets.

The focus of this work is on characterising the degradation in system-wide performance resulting from self-interested decision making. For that purpose, we consider the following standard equilibrium notions.
\begin{definition}[Pure Nash equilibrium]
A profile \(s^{\mathrm{ne}} \in S\) is a pure Nash equilibrium (PNE) if
\begin{equation}\label{eq:pne}
\pi_i(s^{\mathrm{ne}}) \ge \pi_i(s_i', s^{\mathrm{ne}}_{-i})
\qquad
\text{for all } i \in N,\ s_i' \in S_i .
\end{equation}
\end{definition}
We write \(\mathrm{PNE}(G)\) for the set of all PNE of a game~\(G\).

The notion of pure Nash equilibrium involves only deterministic strategy profiles. Broader equilibrium notions like mixed Nash equilibria, correlated equilibria, and coarse correlated equilibria, however, require players to compare lotteries over outcomes. This induces lotteries over pure strategy profiles, that is, probability distributions \(\sigma \in \Delta(S)\).

We evaluate such lotteries in the standard von Neumann and Morgenstern framework \cite{vonneumann1944theory}. If a player's preference over lotteries satisfies the vNM axioms\footnote{Let $\succsim_i$ denote player~$i$'s preference over lotteries in $\Delta(S)$. The vNM axioms are: \emph{Completeness}: for all $\sigma,\tau \in \Delta(S)$, either $\sigma \succsim_i \tau$ or $\tau \succsim_i \sigma$. \emph{Transitivity}: $\sigma \succsim_i \tau$ and $\tau \succsim_i \rho$ imply $\sigma \succsim_i \rho$. \emph{Continuity}: for all $\sigma \succ_i \tau \succ_i \rho$, there exist $\alpha,\beta \in (0,1)$ such that $\alpha\sigma + (1{-}\alpha)\rho \succ_i \tau \succ_i \beta\sigma + (1{-}\beta)\rho$. \emph{Independence}: $\sigma \succsim_i \tau$ if and only if $\alpha\sigma + (1{-}\alpha)\rho \succsim_i \alpha\tau + (1{-}\alpha)\rho$ for all $\alpha \in (0,1]$ and $\rho \in \Delta(S)$.}, then the same function \(u_i\) introduced above ranks lotteries by expected utility,
\begin{equation}
U_i(\sigma) := \mathbb{E}_{s \sim \sigma}[u_i(s)] .
\label{eq:vnm}
\end{equation}
Moreover, this representation is unique only up to a player specific positive affine transformation $u_i \mapsto \alpha_i u_i + \beta_i, \alpha_i > 0$. Thus, once randomized outcomes are admitted, the behaviorally meaningful content of player \(i\)'s payoff is its affine equivalence class rather than its absolute level.

\begin{figure}[tb]
\centering
\begin{minipage}{0.45\columnwidth}
\centering
\textbf{Game \(G\)}
\[
\begin{array}{c|cc}
 & S_1 & S_2\\ \hline
S_1 & (3,3) & (1,4)\\
S_2 & (4,1) & \mathbf{(2,2)}
\end{array}
\]
\end{minipage}\hfill
\begin{minipage}{0.45\columnwidth}
\centering
\textbf{Game \(G'\):} \(u_1' = 3\, u_1\)
\[
\begin{array}{c|cc}
 & S_1 & S_2\\ \hline
S_1 & (9,3) & (3,4)\\
S_2 & (12,1) & \mathbf{(6,2)}
\end{array}
\]
\end{minipage}
\caption{Scaling player~$1$'s payoff by~$\alpha_1{=}3$ does not change the NE~$(S_2,S_2)$, yet the utilitarian PoA shifts from~$6/4{=}3/2$ to~$13/8$. The Nash welfare ratio~$9/4$ is invariant.}
\label{fig:two_matrix_games}
\end{figure}

Passing to surplus utilities removes the additive terms \(\beta_i\). Hence, after subtracting the refer \eqref{eq:surplus}, the admissible vNM transformations act on surpluses as player specific positive scalings \(\pi_i \mapsto \alpha_i \pi_i\). Correspondingly, the expected surplus \(\mathbb{E}[\pi_i(s)] = U_i(\sigma) - u_i(s^0)\) inherits the same scaling. This representation class will be the relevant one for the welfare analysis in the next subsections.

With the preferences over lotteries fixed, the concept of equilibrium extends naturally to randomized play. A distribution \(\sigma \in \Delta(S)\) over strategy profiles captures settings where agents randomise independently, or where an external mediator induces correlation.

\begin{definition}[Coarse correlated equilibrium]
A distribution \(\sigma \in \Delta(S)\) is a coarse correlated equilibrium, CCE, if no player can improve their expected payoff by committing in advance to a fixed action:
\begin{equation}
\mathbb{E}_{s \sim \sigma}[\pi_i(s)]
\ge
\mathbb{E}_{s_{-i} \sim \sigma_{-i}}[\pi_i(s_i', s_{-i})],
\label{eq:cce}
\end{equation}
$\text{for all } i \in N,\ s_i' \in S_i$, where \(\sigma_{-i}\) denotes the marginal of \(\sigma\) on \(S_{-i}\).
\end{definition}

CCE is the most permissive equilibrium concept in the hierarchy
\[
\mathrm{PNE} \subset \mathrm{MNE} \subset \mathrm{CE} \subset \mathrm{CCE},
\]
where a mixed Nash equilibrium, MNE, is a CCE whose distribution is a product \(\sigma = \sigma_1 \otimes \cdots \otimes \sigma_n\), a correlated equilibrium, CE, additionally requires that no player can improve by deviating after observing their recommendation, and every PNE embeds as a degenerate distribution.
Any bound on the worst-case welfare over CCE therefore applies a~fortiori to all smaller equilibrium classes.

\subsection{Social welfare and the welfarist approach}\label{sec:welfarism}

System performance is measured by a SWF~$W : S \to \mathbb{R}$, which assigns a social welfare value to each~$s \in S$.
A system-optimal strategy profile satisfies $s_{\mathrm{opt}} \in \arg\max_{s \in S} W(s)$.
Thus, a utility-maximization game describing the system is represented by the tuple
\begin{equation}\label{eq:game}
    G:= \{N, S, \{\pi_i\}_{i \in N}, W\}
\end{equation}
Throughout, we adopt the \emph{welfarist approach}~\citep{roberts_interpersonal_1980}, as adapted to control applications in \cite{shilov2025welfare}: under mild conditions, the social objective is a function of the individual utilities.
This approach is widely employed in the literature, most commonly through the \emph{utilitarian} specification, i.e. the social cost is defined by summing individual costs $W^{\mathrm{Util}}(s) = \sum_{i=1}^n \pi_i(s)$~\citep{roughgarden2002poa, roughgarden2004bounding, chandan2024methodologies}.

\begin{figure}[h]
    \centering
    \includegraphics[width=\columnwidth]{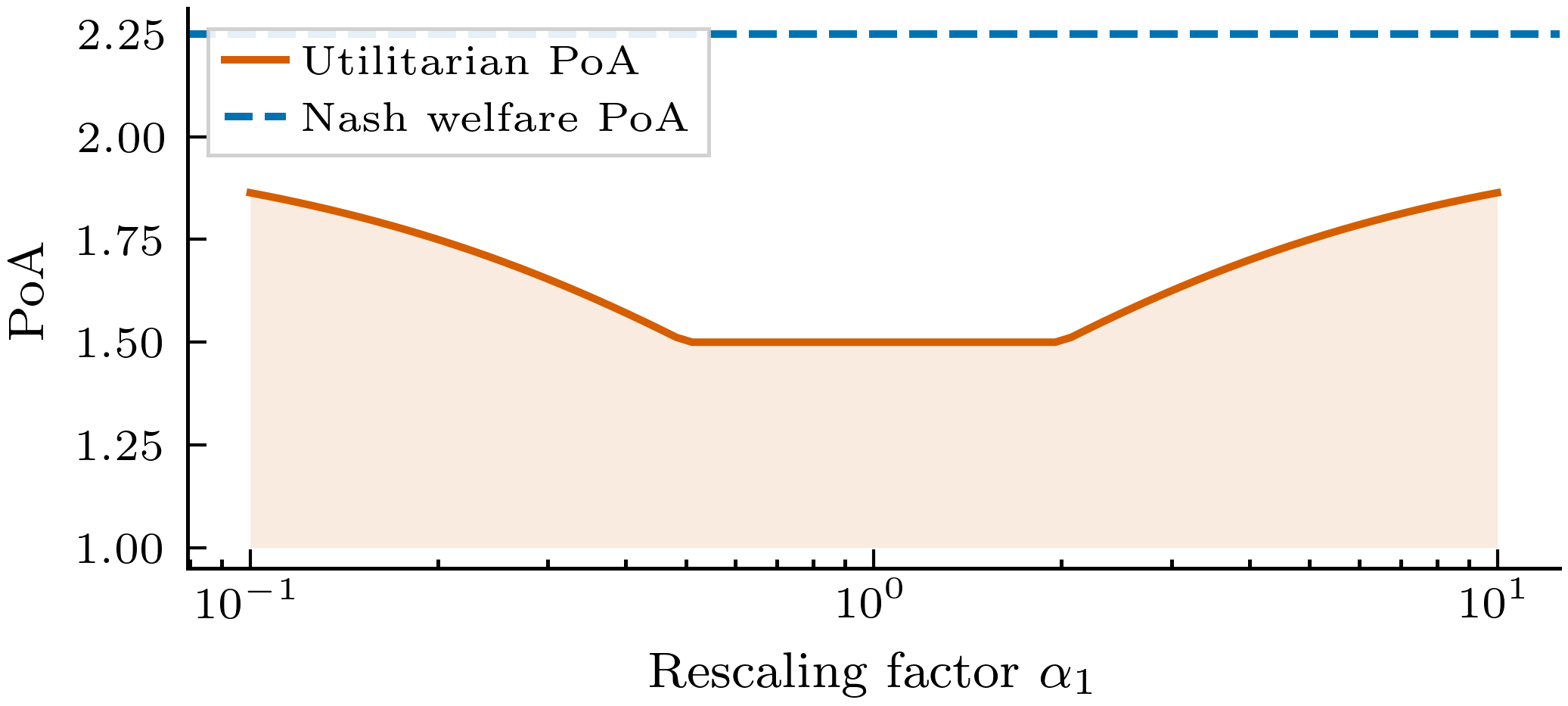}
    \caption{For the game in Figure~\ref{fig:two_matrix_games}, the utilitarian PoA changes as we rescale player~$1$'s utility by~$\alpha_1$, while the Nash welfare PoA remains exactly constant.}
    \label{fig:invariance_motivation}
\end{figure}

However, as discussed in the Introduction, the utilitarian sum is not invariant to representational choices inherent in players' payoff functions. The question is therefore: \emph{which form of $W$ is admissible when payoffs lack full comparability?}

The welfarist approach guides the choice of the aggregation function $W$ consistent with the invariance properties.
Once the individual payoffs are fixed, all information relevant for a social decision is contained in the payoff vector at each outcome.

Given a payoff profile~$\pi(s) = (\pi_1(s),\dots,\pi_n(s))$, a social preference~$\succeq$ ranks strategy profiles~$x, y \in S$.
Classic results show that, under mild conditions, such a ranking is represented by a continuous function $W:\mathbb R^n\to\mathbb R$, and that the admissible form of $W$ is determined by the degree of interpersonal comparability \citep{roberts_interpersonal_1980, dAspremontGevers2002, Sen1979_welfarism, shilov2025welfare}.
For the social preference to be consistent, we first specify which \emph{fundamental properties (or axioms)} it needs to satisfy: Weak Pareto (WP), Partial Independence of Irrelevant Alternatives (PI), and Continuity (C), stated here informally. We refer to~\cite{shilov2025welfare} for a formal treatment and discussion.

\begin{description}
    \item[WP] For any two profiles $x, y \in S$, if $\pi_i(x) > \pi_i(y)$ for all $i\in N$, then $x\succ y$.
    \item[PI] \cite[Sec. 5]{roberts_interpersonal_1980}. The ranking of two profiles does not depend on payoffs at other profiles, except on at most one reference profile.
    \item[C] Small changes in $\pi(x)$ and $\pi(y)$ do not reverse a strict ranking.
\end{description}
An optional property is Anonymity (A), which requires equal treatment of agents who differ only by labels. It encodes fairness, but if predetermined priorities are to be incorporated, anonymity may not be appropriate.
\begin{description}
    \item[A] For any permutation $\tau: N \to N$ of agents, the ranking given by $W$ remains the same.
\end{description}

Under WP, PI, and C there exists a continuous $W:\mathbb R^n\to\mathbb R$ such that
\[
x\succeq y\ \Longleftrightarrow\ W(\pi(x))\,\ge\,W(\pi(y)).
\]
The exact form of $W$ is determined by the chosen level of comparability. In what follows we focus on the cardinal classes that matter for numerical efficiency evaluation, following \cite{roberts_interpersonal_1980}.

\subsection{Comparability classes by invariance} \label{section: Comparability classes by invariance}
The admissible form of $W$ is constrained by the information the designer accepts about interpersonal comparison.
We encode that information through an invariance requirement: the social ranking must be preserved when each player's payoff undergoes an admissible transformation.

\begin{definition}\label{def:invariance}
Let $\Phi$ be a family of strictly increasing real maps.
A list $\varphi=(\varphi_1,\dots,\varphi_n)$ with $\varphi_i\in\Phi$ is an \emph{admissible transformation} if for all $x,y\in S$,
    \begin{equation*}
        W(\pi(x))\ \ge\ W(\pi(y))
         \quad\Leftrightarrow\quad
        W(\pi^{\varphi}(x))\ \ge\ W(\pi^{\varphi}(y)),
    \end{equation*}
    where $\pi_i^{\varphi}(s) = \varphi_i(\pi_i(s))$ is the transformed payoff.
\end{definition}

The choice of $\Phi$ specifies what is admissible for interpersonal comparison and leads to different canonical forms for $W$ \citep{shilov2025welfare}.
We focus on two main cases.

\subsubsection{Cardinal Unit Comparability (CUC) and Cardinal Full Comparability (CFC)}
Here $\Phi_{\mathrm{CUC}} = \big\{\,\varphi_i(t)=a\,t+b_i\ \text{with}\ a>0,\ b_i\in\mathbb R\,\big\}$: the scaling factor~$a$ is \emph{common} across agents, while baselines~$b_i$ may differ (see e.g. Observation 1 in \cite{chandan2024methodologies}). This amounts to assuming that payoff units are commensurable (e.g., all measured in the same currency), though zero-levels need not be aligned. Passing to surpluses removes the offsets $b_i$, so on surplus vectors CUC induces exactly a common positive scaling, $\pi_i \mapsto a\,\pi_i$ for all $i$, that is, the CFC action class. Under CUC (or CFC on surpluses), the admissible aggregators form a one-parameter family parametrised by an inequality aversion parameter~$\rho \ge 0$; the utilitarian sum $W^{\mathrm{Util}}(s) = \sum_i w_i \pi_i(s)$ ($\rho = 0$) is the most prominent member.

However, CUC and CFC are too restrictive in general. They assume a common interpersonal unit of measurement, and standard game theoretic payoff models do not supply such comparability by default. Rather, absent additional structure, each agent's payoff is only defined up to an agent specific positive affine transformation, e.g. as with vNM expected utility. Figure~\ref{fig:two_matrix_games} illustrates the issue: an individual affine change applied to one player leaves the strategic structure unchanged, yet can change the utilitarian PoA. Any aggregator that is admissible only under CUC (CFC) can therefore produce welfare rankings that vary under agent specific rescalings, even when the underlying game is informationally unchanged.

\subsubsection{Cardinal Non Comparability (CNC)}
Here $\Phi_{\mathrm{CNC}}$ contains all agent-specific positive affine maps $\varphi_i(t)=a_i t + b_i$, $a_i>0$. Units and zeros may differ across agents, corresponding to the most general setting, and the one aligned with vNM utility theory. Since surpluses absorb offsets (cf.~\eqref{eq:surplus}), the CNC invariance class acts on surpluses as individual positive scalings $\pi_i \mapsto a_i \pi_i$. Under CNC with $\pi_i(s) > 0$ for all $i$ and $s$, classical social-choice representation results imply that, under the welfarist axioms adopted here, the induced \emph{social welfare ordering} is of weighted Nash type on baseline surpluses~\citep{roberts_interpersonal_1980, shilov2025welfare}:
\begin{equation}\label{eq:nsw}
    W^{\mathrm{Nash}}(\pi(s))\;=\;\prod_{i \in N} \pi_i(s)^{p_i},
\end{equation}
with $p_i > 0$, $\sum_i p_i = 1$, and $p_i = p_j$ for all $i,j$ if
\textbf{A} is imposed. We therefore adopt the weighted Nash product as a
canonical cardinal representative of this ordering. Parameter $p_i$ may reflect the relative importance of agents, e.g. job weights in load balancing, bargaining power in bargaining games etc.

\subsection{Price of Anarchy}\label{sec:price_of_anarchy}

Assuming $\mathrm{NE}(G)$ is nonempty and taking $W = W^{\mathrm{Nash}}$ as in~\eqref{eq:nsw}, we define the \emph{CNC-invariant Price of Anarchy} of the game~$G$ as
\begin{equation}\label{eq:poa}
\poa(G) \;:=\; \frac{\max_{s \in S}\; W^{\mathrm{Nash}}(\pi(s))}{\min_{s \in \mathrm{NE}(G)}\; W^{\mathrm{Nash}}(\pi(s))} \;\ge\; 1.
\end{equation}
The CNC-invariant PoA represents the ratio between the Nash welfare at the social optimum and at the worst-performing Nash equilibrium. Under a CNC transformation $u_i \mapsto \alpha_i u_i + \beta_i$, surpluses transform as $\pi_i \mapsto \alpha_i \pi_i$, and the scaling factors cancel in the welfare ratio~\eqref{eq:poa}, so $\poa(G)$ depends only on the underlying preferences and not on the choice of units or baselines.

For a given class of games~$\mathcal{G}$, which may contain infinitely many instances, the PoA is
\begin{equation}\label{eq:poa_class}
\poa(\mathcal{G}) \;:=\; \sup_{G \in \mathcal{G}}\; \poa(G) \;\ge\; 1.
\end{equation}
A lower $\poa(\mathcal{G})$ corresponds to better worst-case equilibrium performance; $\poa(\mathcal{G}) = 1$ implies that all Nash equilibria of all games in~$\mathcal{G}$ are socially optimal.

\begin{remark}[PoA over broader equilibrium classes]\label{rem:poa_cce}
The definition~\eqref{eq:poa} extends naturally by replacing $\mathrm{NE}(G)$ with the set of MNE, CE, or CCE. Since $\mathrm{NE} \subset \mathrm{MNE} \subset \mathrm{CE} \subset \mathrm{CCE}$, the PoA over CCE is the weakest (largest) bound and implies all others. In the smoothness framework of Section~\ref{sec:smoothness}, the pure Nash bound is obtained under general $(\Lambda,\lambda)$ multiplicative smoothness. For randomized equilibrium classes, we establish the extension for the case $(\Lambda,1)$, which is also the case used in Section~\ref{sec:envelopes}. No-regret learning dynamics converge to the set of CCE~\citep{roughgarden_robustness}, so smoothness-based PoA bounds also characterise the long-run efficiency of decentralised learning.
\end{remark}

\section{Multiplicative Smoothness}\label{sec:smoothness}

Having defined the CNC-invariant PoA in Section~\ref{sec:price_of_anarchy}, we now develop a method for bounding it.
In the classical utilitarian setting, the $(\lambda,\mu)$-smoothness framework~\citep{roughgarden_robustness, chandan2024methodologies} provides worst-case PoA bounds by establishing an \emph{additive} inequality that holds for all pairs of strategy profiles.

\begin{definition}[$(\lambda,\mu)$-smoothness~\citep{roughgarden_robustness, chandan2024methodologies}]\label{def:additive_smooth}
A utility-maximisation game with utilitarian welfare $W^{\mathrm{Util}}(s) = \sum_{i} \pi_i(s)$ is $(\lambda,\mu)$-smooth if $\lambda > 0$, $\mu \ge 0$, and for all $s, s' \in \mathcal{S}$,
\begin{equation}\label{eq:additive_smooth}
\sum_{i=1}^n \pi_i(s_i', s_{-i}) \;\ge\; \lambda\, W^{\mathrm{Util}}(s') \;-\; \mu\, W^{\mathrm{Util}}(s).
\end{equation}
This implies $\poa(G) \le (1 + \mu) / \lambda$ \cite{roughgarden_robustness}.
\end{definition}

The additive structure of~\eqref{eq:additive_smooth} is matched to the additive structure of the utilitarian sum.
However, this structure is not well-defined under CNC. Under an admissible transformation $\pi_i \mapsto \alpha_i \pi_i$ with agent-specific $\alpha_i > 0$, condition~\eqref{eq:additive_smooth} becomes $\sum_i \alpha_i \pi_i(s_i', s_{-i}) \ge \lambda_\alpha\, W^{\mathrm{Util}}_\alpha(s') - \mu_\alpha\, W^{\mathrm{Util}}_\alpha(s)$, where $W^{\mathrm{Util}}_\alpha(s) = \sum_i \alpha_i \pi_i(s)$. Unless all $\alpha_i$ are equal (i.e., CUC rather than CNC), the tightest parameters $(\lambda_\alpha, \mu_\alpha)$ depend on the choice of~$\alpha$, so the resulting PoA bound $(1+\mu_\alpha)/\lambda_\alpha$ is representation-dependent. That motivates us to consider a multiplicative structure instead.

\subsection{Multiplicative structure}\label{sec:positive}

The Nash welfare~\eqref{eq:nsw} is defined as a product of powers $W^{\mathrm{Nash}}(\pi(s)) = \prod_i \pi_i(s)^{p_i}$, and is well-defined for strictly positive arguments. Since surpluses $\pi_i(s) = u_i(s) - u_i(s^0) > 0$ are positive by construction (Section~\ref{sec:framework}), the Nash welfare is well-defined and serves as the social welfare to be maximised.

Since all surpluses are strictly positive, all payoff ratios are well-defined, and ratios turn out to be the \emph{only} invariant primitive available under CNC. Recall that each agent's payoff is determined only up to an individual positive scaling~$\alpha_i$. A statement of the form ``switching from~$s$ to~$s'$ improves agent~$i$'s payoff by~$\delta$'' has no intrinsic meaning, since the magnitude~$\delta$ changes with the choice of units. By contrast, the statement ``switching multiplies agent~$i$'s payoff by a factor of~$r$'' is invariant: if $\pi_i' = \alpha_i \pi_i$, the ratio $\pi_i(s)/\pi_i(s') = \pi_i'(s)/\pi_i'(s')$ is preserved. Thus, under CNC, deviations can only be evaluated in terms of proportional gains or losses.\footnote{Payoff ratios are the \emph{minimal} CNC-invariant primitive. Let $\psi(\pi_i(s), \pi_i(s'))$ be any real-valued function that is invariant under $\pi_i \mapsto \alpha_i \pi_i$, $\alpha_i > 0$. Setting $\alpha_i = 1/\pi_i(s')$ gives $\psi(\pi_i(s)/\pi_i(s'), 1)$, so $\psi$ depends only on the ratio.}

This observation allows us to rewrite the equilibrium conditions of Section~\ref{sec:price_of_anarchy} in ratio form. Since $\pi_i > 0$, the PNE condition~\eqref{eq:pne} is equivalent to
\begin{equation}\label{eq:ne_ratio}
\frac{\pi_i(s^{ne})}{\pi_i(s_i', s_{-i}^{ne})} \ge 1 \quad \text{for all } i \in \mathcal{N}, \; s_i' \in S_i,
\end{equation}
and the CCE condition~\eqref{eq:cce} is equivalent to
\begin{equation}\label{eq:mne_ratio}
\frac{\E_{s \sim \sigma}[\pi_i(s)]}{\E_{s \sim \sigma}[\pi_i(s_i', s_{-i})]} \ge 1 \quad \text{for all } i \in \mathcal{N}, \; s_i' \in S_i.
\end{equation}
Both formulations are CNC-invariant: the scaling factor $\alpha_i$ cancels in the ratio. More importantly, they reveal the \emph{multiplicative} structure of equilibrium: taking the product over agents with exponents~$p_i$, at the PNE we have
\begin{equation}\label{eq:ne_product}
W^{\mathrm{Nash}}(\pi(s^{ne})) \;\ge\; \prod_{i=1}^n \pi_i(s_i', s_{-i}^{ne})^{p_i} \quad \text{for all } s' \in \mathcal{S}.
\end{equation}
That is, the Nash welfare at equilibrium is bounded below by the weighted geometric mean of deviation payoffs. This product-of-deviations structure is precisely what we must control in order to bound the PoA, which calls for a smoothness condition built from products rather than sums.

\subsection{Multiplicative smoothness}\label{sec:ms}

With the ratio-based characterization~\eqref{eq:ne_product} in hand, we are ready to state the smoothness condition matched to the product structure of Nash welfare.

\begin{definition}[Multiplicative smoothness]\label{def:ms}
A game $G$ is $(\Lambda,\lambda)$-\emph{multiplicatively smooth} if $\Lambda \ge 1$, $\lambda \in (0,1]$, and for all $s, s' \in \mathcal{S}$,
\begin{equation}\label{eq:ms}
\prod_{i=1}^n \pi_i(s_i', s_{-i})^{p_i}
\ge
\frac{1}{\Lambda}\, W^{\mathrm{Nash}}(\pi(s'))^{\lambda}\, W^{\mathrm{Nash}}(\pi(s))^{1-\lambda}.
\end{equation}
\end{definition}

The left-hand side is the weighted geometric mean of the payoffs that agents would obtain if each unilaterally deviated from~$s$ to their action in~$s'$. The right-hand side interpolates between the welfare at $s'$ (weighted by~$\lambda$) and at $s$ (weighted by~$1{-}\lambda$), with $\Lambda$ measuring how far deviation payoffs can fall below this interpolation.

The exponents $\lambda$ and $1{-}\lambda$ sum to~$1$, so \eqref{eq:ms} is scale-invariant: under $\pi_i \mapsto \alpha_i \pi_i$, both sides scale by $\prod_i \alpha_i^{p_i}$. If they did not sum to~$1$, rescaling would break the inequality, violating CNC invariance.

\begin{theorem}[Pure NE bound]\label{thm:pure}
If the game $G$ is $(\Lambda,\lambda)$-multiplicatively smooth, then
\begin{equation}\label{eq:pure_bound}
\poa(G) \le \Lambda^{1/\lambda}.
\end{equation}
\end{theorem}

\begin{proof}
Let $s^\star \in \argmax_{s} W^{\mathrm{Nash}}(\pi(s))$. At any PNE $s^{ne}$, the product form~\eqref{eq:ne_product} gives $W^{\mathrm{Nash}}(\pi(s^{ne})) \ge \prod_i \pi_i(s_i^\star, s_{-i}^{ne})^{p_i}$. Applying~\eqref{eq:ms} with $s = s^{ne}$, $s' = s^\star$:
\[
W^{\mathrm{Nash}}(\pi(s^{ne}))
\;\ge\;
\frac{1}{\Lambda}\, W^{\mathrm{Nash}}(\pi(s^\star))^{\lambda}\, W^{\mathrm{Nash}}(\pi(s^{ne}))^{1-\lambda}.
\]
Dividing by $W^{\mathrm{Nash}}(\pi(s^{ne}))^{1-\lambda}$ and raising to $1/\lambda$ yields the result.
\end{proof}

\emph{Extension to coarse correlated equilibria.}
Theorem~\ref{thm:pure} is stated for general $(\Lambda,\lambda)$, with $\lambda$ interpolating between welfare at $s'$ and at $s$. For the bounds derived in Section~\ref{sec:envelopes}, we use $\lambda = 1$, for which we also derive the CCE extension.
For a CCE~$\sigma \in \Delta(\mathcal{S})$, we evaluate welfare at the expected payoff vector $\bar\pi(\sigma) \;:=\; \bigl(\E_{s \sim \sigma}[\pi_1(s)],\; \dots,\; \E_{s \sim \sigma}[\pi_n(s)]\bigr)$, and define $\poa_{\mathrm{CCE}}(G) := \max_{s} W^{\mathrm{Nash}}(\pi(s))\, /\, \min_{\sigma \in \mathrm{CCE}(G)} W^{\mathrm{Nash}}(\bar\pi(\sigma))$. As Nash SWF is non-linear, the choice between $(W(E[\pi])) \text{ and } (E[W(\pi)])$ is substantive. We focus on evaluation of the lotteries by the ex ante Nash SWF, evaluated on the expected surplus vector. This choice keeps the social evaluation on the same expected utilities that define CCE and does not impose an additional ex post social aggregation over realized profiles.

\begin{theorem}[CCE bound]\label{thm:mne}
If the game is $(\Lambda,1)$-multiplicatively smooth, i.e.,
\begin{equation}\label{eq:ms_lambda1}
\prod_{i=1}^n \pi_i(s_i', s_{-i})^{p_i}
\ge
\frac{1}{\Lambda}\, W^{\mathrm{Nash}}(\pi(s')) \quad \text{for all } s, s' \in \mathcal{S},
\end{equation}
then $\poa_{\mathrm{CCE}}(G) \le \Lambda$.
\end{theorem}

\begin{proof}
Let $s^\star \in \argmax_s W^{\mathrm{Nash}}(\pi(s))$. The CCE ratio condition~\eqref{eq:mne_ratio} gives $\E[\pi_i(s)] \ge \E[\pi_i(s_i^\star, s_{-i})]$ for all~$i$. Taking products with exponents~$p_i$:
\[
W^{\mathrm{Nash}}(\bar\pi(\sigma)) = \prod_i \E[\pi_i(s)]^{p_i} \ge \prod_i \E[\pi_i(s_i^\star, s_{-i})]^{p_i}.
\]
By the weighted H\"older inequality,
$\prod_i \E[\pi_i(s_i^\star, s_{-i})]^{p_i} \ge \E[\prod_i \pi_i(s_i^\star, s_{-i})^{p_i}]$.
By~\eqref{eq:ms_lambda1}, the integrand is at least $\Lambda^{-1} W^{\mathrm{Nash}}(\pi(s^\star))$ pointwise. Taking expectations gives $W^{\mathrm{Nash}}(\bar\pi(\sigma)) \ge \Lambda^{-1} W^{\mathrm{Nash}}(\pi(s^\star))$, hence $\poa_{\mathrm{CCE}}(G) \le \Lambda$.
\end{proof}

\section{Envelope Based Bounds via Retention Envelopes}\label{sec:envelopes}

To apply the smoothness framework of Section~\ref{sec:smoothness}, one must verify condition~\eqref{eq:ms} for specific game classes. The key challenge is bounding a product of payoff ratios over all possible deviations, which is an ostensibly high-dimensional problem. We reduce it to a one-dimensional problem by introducing \emph{retention envelopes}: a single function that captures the worst-case multiplicative retention of utility when the aggregate weight on a resource increases, regardless of the underlying utility function.

\subsection{Weighted single-choice welfare games}\label{sec:lb_model}

Consider a finite set of resources~$\mathcal{M}$ with positive, nonincreasing utility functions $f_j: \mathbb{R}_{>0} \to \mathbb{R}_{>0}$, and $n$~players with weights $w_i > 0$. Each player chooses one resource; a strategy profile $s = (s_1, \dots, s_n)$ with $s_i \in \mathcal{M}$ induces aggregate weight $k_j(s) = \sum_{i:\, s_i = j} w_i$. The payoff of player~$i$ is the utility on the chosen resource: $\pi_i(s) = f_{s_i}(k_{s_i}(s))$.
This is the utility-maximisation analogue of weighted load balancing, with latency functions replaced by utility functions. The model is an instance of single-choice distributed welfare games~\citep{marden2013distributed}---the utility-maximization analogue of congestion games, which are widely used in the context of resource allocation. We set $W_{\mathrm{tot}} = \sum_i w_i$, $p_i = w_i / W_{\mathrm{tot}}$, and evaluate profiles using the weighted Nash welfare $W^{\mathrm{Nash}}(\pi(s)) = \prod_i \pi_i(s)^{p_i}$ as in Section~\ref{sec:price_of_anarchy}.

\subsection{Retention envelope and geometric closure}\label{sec:envelope_def}

The key idea is to reduce the multi-dimensional smoothness
verification to a one-dimensional worst-case calculation. Let
$\mathcal{F}$ be a set of positive, nonincreasing functions.

\begin{definition}[Retention envelope]\label{def:envelope}
For $u \ge 0$, the retention envelope of $\mathcal{F}$ is
\[
\psi_{\mathcal{F}}(u)
\;:=\;
\inf_{f\in\mathcal{F}}\ \inf_{x>0}\ \frac{f((1+u)x)}{f(x)}.
\]
The function $\psi_{\mathcal{F}}(u)$ captures the worst-case multiplicative
retention of utility when the aggregate weight on a resource is scaled by
a factor $(1+u)$.
\end{definition}

Since every $f\in\mathcal{F}$ is nonincreasing, the ratio
$u \mapsto {f((1+u)x)}/{f(x)}$
is nonincreasing for each fixed $f\in\mathcal{F}$ and $x>0$. Hence
$\psi_{\mathcal{F}}$ is nonincreasing on $[0,\infty)$.

\begin{definition}[Geometric closure]\label{def:gc}
Let $\psi:[0,\infty)\to(0,\infty)$. The geometric closure of $\psi$ is
\begin{multline*}
\psi^{\mathrm{gc}}(x)
\;:=\;
\inf \Big\{
\prod_{m=1}^M \psi(u_m)^{\alpha_m}
\;:\;
M\in\mathbb{N},\ \alpha_m\ge 0,\\
\sum_{m=1}^M \alpha_m = 1,\
\sum_{m=1}^M \alpha_m u_m = x
\Big\}.
\end{multline*}
Thus, $\psi^{\mathrm{gc}}(x)$ is the smallest weighted geometric mean of
values of $\psi$ compatible with an average expansion factor equal to $x$.
\end{definition}

\begin{lemma}\label{lem:gc_monotone}
If $\psi$ is nonincreasing, then so is $\psi^{\mathrm{gc}}$.
\end{lemma}

\begin{proof}
Fix $0\le x\le y$. Let
$x=\sum_{m=1}^M \alpha_m u_m$
be any representation with $\alpha_m\ge 0$, $\sum_m \alpha_m=1$.
Define $v_m := u_m + (y-x)$ for $m=1,\dots,M$.
Then $\sum_{m=1}^M \alpha_m v_m = y$.
Since $\psi$ is nonincreasing and $v_m\ge u_m$ for all $m$,
\[
\prod_{m=1}^M \psi(v_m)^{\alpha_m}
\;\le\;
\prod_{m=1}^M \psi(u_m)^{\alpha_m}.
\]
Taking the infimum over all representations of $y$ on the left, and then
over all representations of $x$ on the right, yields
$\psi^{\mathrm{gc}}(y)\le \psi^{\mathrm{gc}}(x)$.
\end{proof}

\begin{figure}[t]
\centering
\begin{tikzpicture}[
    >=Latex,
    resource/.style={draw, rounded corners=2pt,
        minimum width=2.1cm, minimum height=0.72cm, thick, inner sep=0pt},
    playerbox/.style={draw, rounded corners=1pt,
        minimum width=0.32cm, minimum height=0.28cm,
        thick, fill=gray!10, inner sep=0pt},
]

\coordinate (c1) at (0,0);
\coordinate (c2) at (2.7,0);
\coordinate (c3) at (5.4,0);

\def\toprow{1.20}
\def\botrow{-0.85}
\def\titlerow{1.95}

\node[font=\footnotesize] at ($(c1)+(0,\titlerow)$)  {resource $1$};
\node[font=\footnotesize] at ($(c2)+(0,\titlerow)$)  {resource $j$};
\node[font=\footnotesize] at ($(c3)+(0,\titlerow)$)  {resource $m$};

\node[font=\footnotesize, anchor=east] at (-1.15,\toprow) {profile $s$};
\node[font=\footnotesize, anchor=east] at (-1.15,\botrow) {profile $s'$};

\node[resource] (r1t) at ($(c1)+(0,\toprow)$) {};
\node[resource] (rjt) at ($(c2)+(0,\toprow)$) {};
\node[resource] (rmt) at ($(c3)+(0,\toprow)$) {};

\node[resource] (r1b) at ($(c1)+(0,\botrow)$) {};
\node[resource] (rjb) at ($(c2)+(0,\botrow)$) {};
\node[resource] (rmb) at ($(c3)+(0,\botrow)$) {};

\fill[gray!25, rounded corners=1pt]
    ($(r1t.south west)+(0.08,0.08)$) rectangle ($(r1t.south west)+(0.62,0.64)$);
\node[font=\footnotesize] at ($(r1t.south west)+(0.35,0.36)$) {$k_1$};

\fill[gray!40, rounded corners=1pt]
    ($(rjt.south west)+(0.08,0.08)$) rectangle ($(rjt.south west)+(1.15,0.64)$);
\node[font=\footnotesize] at ($(rjt.south west)+(0.60,0.36)$) {$k_j$};

\fill[gray!25, rounded corners=1pt]
    ($(rmt.south west)+(0.08,0.08)$) rectangle ($(rmt.south west)+(0.82,0.64)$);
\node[font=\footnotesize] at ($(rmt.south west)+(0.45,0.36)$) {$k_m$};

\node[playerbox] at ($(r1b.west)+(0.30,0)$) {};
\node[playerbox] at ($(r1b.west)+(0.70,0)$) {};
\node[font=\footnotesize, anchor=west] at ($(r1b.west)+(1.00,0)$) {$o_1$};

\node[playerbox] at ($(rjb.west)+(0.25,0)$) {};
\node[playerbox] at ($(rjb.west)+(0.62,0)$) {};
\node[playerbox] at ($(rjb.west)+(0.99,0)$) {};
\node[playerbox] at ($(rjb.west)+(1.36,0)$) {};
\node[font=\footnotesize, anchor=west] at ($(rjb.west)+(1.60,0)$) {$o_j$};

\node[font=\scriptsize] at ($(rjb.west)+(0.25,-0.52)$) {$w_{i_1}$};
\node[font=\scriptsize] at ($(rjb.west)+(0.62,-0.52)$) {$w_{i_2}$};
\node[font=\scriptsize] at ($(rjb.west)+(0.99,-0.52)$) {$\cdots$};
\node[font=\scriptsize] at ($(rjb.west)+(1.36,-0.52)$) {$w_{i_r}$};

\node[playerbox] at ($(rmb.west)+(0.30,0)$) {};
\node[playerbox] at ($(rmb.west)+(0.70,0)$) {};
\node[playerbox] at ($(rmb.west)+(1.10,0)$) {};
\node[font=\footnotesize, anchor=west] at ($(rmb.west)+(1.38,0)$) {$o_m$};

\draw[->, thick, gray!70!black]
    ($(rjb.north)+(0,0.08)$) -- ($(rjt.south)+(0,-0.08)$);

\node[font=\footnotesize] at ($(c1)!0.5!(c2)+(0,\toprow)$) {$\cdots$};
\node[font=\footnotesize] at ($(c2)!0.5!(c3)+(0,\toprow)$) {$\cdots$};
\node[font=\footnotesize] at ($(c1)!0.5!(c2)+(0,\botrow)$) {$\cdots$};
\node[font=\footnotesize] at ($(c2)!0.5!(c3)+(0,\botrow)$) {$\cdots$};

\node[anchor=north, align=center, font=\footnotesize]
    at ($(c2)+(0,-1.75)$)
    {$f_j(k_j{+}w_i) \;\ge\; f_j(k_j{+}o_j)$
     \;\;since $o_j \ge w_i$, $f_j{\downarrow}$};

\end{tikzpicture}
\caption{Grouping deviation terms by target resource.
Each player~$i$ with $s_i'{=}j$ adds weight $w_i \le o_j$ to resource~$j$.
Since $f_j$ is nonincreasing, the deviation payoff is bounded below by
$f_j(k_j{+}o_j)$, converting the player-wise product
into~\eqref{eq:dev_bound}.}
\label{fig:resource-grouping}
\end{figure}
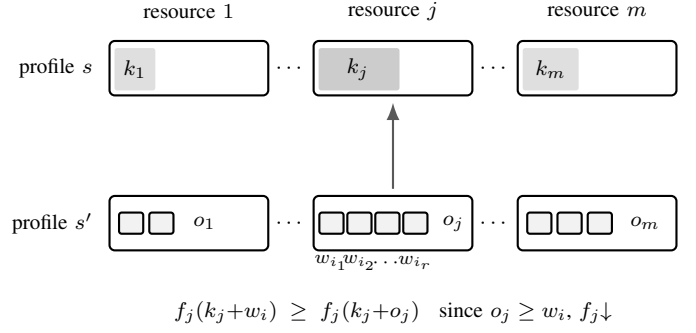

\subsection{Envelope bound}\label{sec:envelope_thm}

\begin{theorem}[Envelope bound for weighted single-choice welfare games]\label{thm:envLB}
Assume each $f_j \in \mathcal{F}$ is nonincreasing. Let
$\Lambda_{\mathcal{F}} := {1}/{\psi_{\mathcal{F}}^{\mathrm{gc}}(1)}$.
Then for all profiles $s, s'$,
\begin{equation}\label{eq:smooth_lb}
\prod_{i=1}^n \pi_i(s_i', s_{-i})^{w_i}
\;\ge\;
\Lambda_{\mathcal{F}}^{-W_{\mathrm{tot}}}
\prod_{i=1}^n \pi_i(s')^{w_i},
\end{equation}
and every Nash equilibrium satisfies $\poa(G) \le \Lambda_{\mathcal{F}}$.
\end{theorem}

\begin{proof}
Fix profiles $s, s'$. Write $k_j := k_j(s)$ and $o_j := k_j(s')$ for the aggregate weights under $s$ and $s'$, respectively.

Consider any resource~$j$ and a player~$i$ with $s_i' = j$. In the unilateral deviation $(s_i', s_{-i})$, player~$i$ adds weight~$w_i$ to resource~$j$, giving aggregate weight $k_j + w_i$. Since $o_j = \sum_{i':\, s_{i'}' = j} w_{i'} \ge w_i$ and $f_j$ is nonincreasing,
\[
\pi_i(s_i', s_{-i}) = f_j(k_j + w_i) \;\ge\; f_j(k_j + o_j).
\]
Raise to the power $w_i$ and multiply over all players, grouping by the target resource in~$s'$:
\begin{equation}\label{eq:dev_bound}
\prod_{i=1}^n \pi_i(s_i', s_{-i})^{w_i} \;\ge\; \prod_{j \in \mathcal{M}} f_j(k_j + o_j)^{o_j}.
\end{equation}

Since every player assigned to resource $j$ in $s'$ receives payoff $f_j(o_j)$,
we also have
$\prod_{i=1}^n \pi_i(s')^{w_i} = \prod_{j \in \mathcal{M}} f_j(o_j)^{o_j}$.
Therefore,
\[
\frac{\prod_{i=1}^n \pi_i(s_i', s_{-i})^{w_i}}
     {\prod_{i=1}^n \pi_i(s')^{w_i}}
\;\ge\;
\prod_{j:\,o_j>0}
\left(
\frac{f_j(k_j+o_j)}{f_j(o_j)}
\right)^{o_j}.
\]
For each $j$ with $o_j>0$, define
$u_j := {k_j}/{o_j}$ and $\theta_j := {o_j}/{W_{\mathrm{tot}}}$.
Then $k_j+o_j=(1+u_j)o_j$ and $\sum_{j:\,o_j>0}\theta_j = 1$.
Hence, after taking the $1/W_{\mathrm{tot}}$ power,
\begin{equation}\label{eq:ratio_prod}
\left(
\frac{\prod_{i=1}^n \pi_i(s_i', s_{-i})^{w_i}}
     {\prod_{i=1}^n \pi_i(s')^{w_i}}
\right)^{\!1/W_{\mathrm{tot}}}
\;\ge\;
\prod_{j:\,o_j>0}
\left(
\frac{f_j((1+u_j)o_j)}{f_j(o_j)}
\right)^{\!\theta_j}\!.
\end{equation}

By definition of the retention envelope,
$f_j((1+u_j)o_j)/f_j(o_j) \ge \psi_{\mathcal{F}}(u_j)$
for all $j$ with $o_j>0$.
Substituting into~\eqref{eq:ratio_prod} gives
\begin{equation}\label{eq:psi_prod}
\left(
\frac{\prod_{i=1}^n \pi_i(s_i', s_{-i})^{w_i}}
     {\prod_{i=1}^n \pi_i(s')^{w_i}}
\right)^{\!1/W_{\mathrm{tot}}}
\;\ge\;
\prod_{j:\,o_j>0}\psi_{\mathcal{F}}(u_j)^{\theta_j}.
\end{equation}

Now define
$\bar u := \sum_{j:\,o_j>0}\theta_j u_j
= {\sum_{j:\,o_j>0}k_j}/{W_{\mathrm{tot}}} \le 1$,
where the last inequality uses $\sum_{j\in\mathcal{M}} k_j = W_{\mathrm{tot}}$.
Since the coefficients $(\theta_j)_{j:\,o_j>0}$ form a convex combination
with average $\bar u$, the definition of the geometric closure yields
\[
\prod_{j:\,o_j>0}\psi_{\mathcal{F}}(u_j)^{\theta_j}
\;\ge\;
\psi_{\mathcal{F}}^{\mathrm{gc}}(\bar u).
\]
Since $\psi_{\mathcal{F}}$ is nonincreasing, Lemma~\ref{lem:gc_monotone} implies that
$\psi_{\mathcal{F}}^{\mathrm{gc}}$ is nonincreasing as well, and thus
$\psi_{\mathcal{F}}^{\mathrm{gc}}(\bar u) \ge \psi_{\mathcal{F}}^{\mathrm{gc}}(1)$.
Combining with~\eqref{eq:psi_prod}, we obtain
\[
\left(
\frac{\prod_{i=1}^n \pi_i(s_i', s_{-i})^{w_i}}
     {\prod_{i=1}^n \pi_i(s')^{w_i}}
\right)^{\!1/W_{\mathrm{tot}}}
\;\ge\;
\psi_{\mathcal{F}}^{\mathrm{gc}}(1)
=
\Lambda_{\mathcal{F}}^{-1}.
\]
Raising both sides to the power $W_{\mathrm{tot}}$ yields~\eqref{eq:smooth_lb}.

Let $s^\star \in \argmax_s W(\pi(s))$. At any NE~$s^{\mathrm{ne}}$, $\pi_i(s^{\mathrm{ne}}) \ge \pi_i(s_i^\star, s_{-i}^{\mathrm{ne}})$ for all~$i$. Raising to $w_i$, multiplying, and applying~\eqref{eq:smooth_lb} with $s = s^{\mathrm{ne}}$, $s' = s^\star$:
\[
\prod_i \pi_i(s^{\mathrm{ne}})^{w_i}
\;\ge\;
\prod_i \pi_i(s_i^\star, s_{-i}^{\mathrm{ne}})^{w_i}
\;\ge\;
\Lambda_{\mathcal{F}}^{-W_{\mathrm{tot}}} \prod_i \pi_i(s^\star)^{w_i}.
\]
Taking the $1/W_{\mathrm{tot}}$-th power gives $W(\pi(s^{\mathrm{ne}})) \ge \Lambda_{\mathcal{F}}^{-1}\, W(\pi(s^\star))$.
\end{proof}

By Theorem~\ref{thm:mne}, the same $(\Lambda_{\mathcal{F}},1)$-smoothness bound also extends to coarse correlated equilibria.

\subsection{Power law utilities}\label{sec:poly}

We now specialise to power law utility curves, obtaining a closed-form envelope bound as a corollary.

\begin{definition}[Power law family]\label{def:poly}
Let $\mathcal{P}^{\downarrow}(p)$ denote the family of functions $f(x) = \sum_{d=0}^p \alpha_d\, x^{-d}$ with $\alpha_d \ge 0$.
\end{definition}

\begin{lemma}[Power law scaling]\label{lem:poly}
If $f \in \mathcal{P}^{\downarrow}(p)$ and $t \ge 1$, then $f(t\, x) \ge t^{-p}\, f(x)$ for all $x > 0$.
\end{lemma}

\begin{proof}
$f(t\, x) = \sum_{d=0}^p \alpha_d\, t^{-d}\, x^{-d} \ge \sum_{d=0}^p \alpha_d\, t^{-p}\, x^{-d} = t^{-p}\, f(x)$, since $t^{-d} \ge t^{-p}$ for $d \le p$ when $t \ge 1$.
\end{proof}

\begin{corollary}\label{cor:poly}
For $\mathcal{F} = \mathcal{P}^{\downarrow}(p)$, $\Lambda_{\mathcal{F}} = 2^p$. Consequently, every Nash equilibrium and every coarse correlated equilibrium of a weighted single-choice welfare game with power law utilities of degree at most~$p$ satisfies $\poa(G) \le 2^p$.
\end{corollary}

\begin{proof}
Lemma~\ref{lem:poly} gives the lower bound $\psi_{\mathcal{F}}(u) \ge (1+u)^{-p}$.
Equality is attained by the function $f(x) = x^{-p}$, for which
$f((1+u)x)/f(x) = (1+u)^{-p}$ for all $x>0$.
Hence $\psi_{\mathcal{F}}(u) = (1+u)^{-p}$.

Now let $x=1$. For any representation
$1=\sum_{m=1}^M \alpha_m u_m$ with $\alpha_m\ge 0$ and $\sum_{m=1}^M \alpha_m=1$,
\[
\prod_{m=1}^M \psi_{\mathcal{F}}(u_m)^{\alpha_m}
=
\prod_{m=1}^M (1+u_m)^{-p\alpha_m}
=
\left(\prod_{m=1}^M (1+u_m)^{\alpha_m}\right)^{\!-p}.
\]
By the weighted AM--GM inequality,
\[
\prod_{m=1}^M (1+u_m)^{\alpha_m}
\le
\sum_{m=1}^M \alpha_m (1+u_m)
=
1+\sum_{m=1}^M \alpha_m u_m
=
2.
\]
Thus $\prod_{m=1}^M \psi_{\mathcal{F}}(u_m)^{\alpha_m} \ge 2^{-p}$.
Taking the infimum over all such representations yields
$\psi_{\mathcal{F}}^{\mathrm{gc}}(1)\ge 2^{-p}$.
On the other hand, equality is attained by the single-point representation
$M=1$, $\alpha_1=1$, $u_1=1$. Hence
$\psi_{\mathcal{F}}^{\mathrm{gc}}(1)=2^{-p}$
and therefore $\Lambda_{\mathcal{F}} = 1/\psi_{\mathcal{F}}^{\mathrm{gc}}(1) = 2^p$.
The claim now follows from Theorem~\ref{thm:envLB} and Theorem~\ref{thm:mne}.
\end{proof}

\begin{remark}
The bound above can be viewed in relation to the Nash social welfare analysis of selfish load balancing by Bil\`o et al.~\cite{bilo2022nash}, who obtain polynomial bounds in the cost minimisation setting. Our contribution uses a separate smoothness based derivation for the utility maximisation analogue. The same envelope calculation can also be used to derive the corresponding pure Nash bound for load balancing to obtain the same bound $2^p$. However, the utility formulation has an additional advantage in the present framework: once the envelope bound is established, Theorem~\ref{thm:mne} yields the same guarantee for coarse correlated equilibria as well.
\end{remark}

\section{Conclusion}\label{sec:conclusion}

We addressed the representation-dependence of the standard PoA under classical vNM utility theory, where payoffs are unique only up to agent-specific positive affine transformations. Drawing on the welfarist approach from social choice theory, the weighted Nash welfare is the canonical aggregator compatible with Cardinal Non-Comparability (CNC), and we developed \emph{multiplicative smoothness} as the matched bounding methodology for PoA in this quite general case. The resulting framework yields general CNC-invariant PoA bounds for pure Nash with $(\Lambda,\lambda)$ and coarse correlated equilibria bound with $(\Lambda,1)$. Since surpluses absorb offsets, the CNC invariance class acts on surpluses as individual positive scalings, under which the multiplicative smoothness condition~\eqref{eq:ms} is invariant by construction. The retention envelope and geometric closure technique then provides a computationally tractable method for deriving the smoothness constant~$\Lambda$ across arbitrary utility families. For weighted single-choice welfare games with power law utilities of degree at most~$p$, the bound~$\Lambda_{\mathcal{F}} = 2^p$ is obtained via a simple envelope calculation.

More broadly, we argue that game theory is intrinsically robust to how individual agents calibrate their payoffs: no assumption on interpersonal comparability is needed to predict strategic behavior. Any comparability assumption layered on top is not demanded by the strategic model itself but introduced by the designer's choice of welfare criterion. The strength of the resulting efficiency guarantee is therefore bounded by the strength of that assumption. The framework developed here shows that meaningful PoA bounds and corresponding methods can be obtained at the most general level of comparability, preserving the full robustness that equilibrium theory provides.

Several extensions are natural. First, the framework can be extended to broader classes of games, including general distributed welfare games~\citep{marden2013distributed} and cost-minimisation settings, by aligning the Nash welfare benchmark with cost-based payoff structures. Second, the relationship between multiplicative smoothness and the generalised smoothness methodology of~\citet{chandan2024methodologies} deserves further study. In networked systems, developing CNC-invariant mechanism and toll design, where tolls are designed to optimise Nash welfare rather than the utilitarian sum, is a promising direction for more robust allocation policies.
Generally, both the structures of the game and of utilities deserve further inspections and the learning might be that known PoA results relying on cardinal full-comparability may often not apply for the purpose of many applications.

\printbibliography

\end{document}